\documentclass[aps, pra, nofootinbib, twocolumn, 11pt, tightenlines,
notitlepage, superscriptaddress, longbibliography]{revtex4-1}

\usepackage{amsmath}
\usepackage{amsthm}
\usepackage{amssymb}
\usepackage{amsfonts}
\usepackage{mathrsfs}
\usepackage{pgfplots}
\usepackage{xparse} 
\usepackage{graphicx}
\usepackage{color}
\usepackage{hyperref}
\usepackage{qcircuit}
\usepackage[normalem]{ulem}

\hypersetup{
	colorlinks,
	linkcolor={blue},
	citecolor={blue},
	urlcolor={blue}
}

\usepackage{cleveref}

\newcommand{\tV}{\vert\kern-0.25ex\vert\kern-0.25ex\vert}

\newcommand{\ct}{\ensuremath{^{\dagger}}}
\newcommand{\mc}{\mathcal}
\newcommand{\bb}{\mathbb}

\DeclareMathOperator{\tr}{Tr}

\theoremstyle{plain}
\newtheorem{thm}{Theorem}
\newtheorem{lem}[thm]{Lemma}

\theoremstyle{definition}

\theoremstyle{remark}

\definecolor{nblue}{rgb}{0.2,0.2,0.7}
\definecolor{ngreen}{rgb}{0.1,0.5,0.1}
\definecolor{nred}{rgb}{0.8,0.2,0.2}
\definecolor{nblack}{rgb}{0,0,0}

\newcommand{\hide}[1]{}

\newcommand{\ini}{\ensuremath{_\mathrm{in}}}
\newcommand{\out}{\ensuremath{_\mathrm{out}}}

\hypersetup{pdftitle={Coherence in quantum error-correcting codes},
pdfauthor={Joel J. Wallman}}

\begin{document}

\title{Quantum error correction decoheres noise}
\author{Stefanie J. Beale}
\email{sbeale@uwaterloo.ca}
\affiliation{ Institute for Quantum Computing and Department of Physics and
Astronomy, University of Waterloo, Waterloo, Ontario, N2L 3G1, Canada}

\author{Joel J. Wallman}
\email{jwallman@uwaterloo.ca}
\affiliation{ Institute
for Quantum Computing and Department of Applied Mathematics, University of Waterloo, Waterloo,
Ontario, N2L 3G1, Canada}

\author{Mauricio Guti\'errez}
\affiliation{ Department of Physics, College of Science, Swansea University, Singleton Park, Swansea SA2 8PP, United Kingdom}
\affiliation{Escuela de Qu\'imica, Universidad de Costa Rica, San Jos\'e, 2060 Costa Rica}

\author{Kenneth R. Brown}
\affiliation{Department of Electrical and Computer Engineering, Chemistry, and Physics, Duke University, Durham, NC 27708}

\author{Raymond Laflamme}
\affiliation{ Institute for Quantum Computing and Department of Physics and
	Astronomy, University of Waterloo, Waterloo, Ontario, N2L 3G1, Canada}
\affiliation{Perimeter Institute for Theoretical Physics, Waterloo N2L 2Y5 ON, Canada}

\begin{abstract}
Typical studies of quantum error correction assume probabilistic Pauli noise, largely because it is relatively easy to
analyze and simulate. Consequently, the effective logical noise due to
physically realistic coherent errors is relatively unknown. Here, we prove that encoding a system in a stabilizer code and measuring error syndromes decoheres errors, that is, causes coherent errors to converge toward probabilistic Pauli errors, even when no recovery operations are applied. Two practical consequences
are that the error rate in a logical circuit is well quantified by the average
gate fidelity at the logical level and that essentially optimal recovery
operators can be determined by independently optimizing the logical fidelity of
the effective noise per syndrome.
\end{abstract}


\maketitle

\section{Introduction}

Quantum computers are likely to dramatically outperform classical computers,
provided that errors can be corrected enough to make the output reliable. Errors
in a quantum computer can take many forms with differing impacts on an
error-correction procedure. Most studies of the performance of quantum
error-correcting codes only consider probabilistic Pauli errors because they are
easy to simulate via the Gottesman-Knill theorem~\cite{Aaronson2004}. However,
in real systems, it is likely that other noise will also be present.

Determining the performance of an error-correcting code at the logical level
under general noise is complicated because such noise is harder to simulate.
Previous approaches have expanded the class of errors to some larger class that
can still be efficiently simulated~\cite{Gutierrez2015}, performed full
density-matrix simulations~\cite{Gutierrez2016}, used tensor network
descriptions of specific codes~\cite{Darmawan2017,Bravyi2017} or effective
logical process matrices~\cite{Rahn2002, Fern2006, Chamberland2017}. These
methods are suboptimal because they either require a huge amount of resources
to simulate or are indirect approximations. They also do not easily give
structural insight because extrapolating the effective logical noise from the
description of the encoded state is difficult and determining the scaling with
parameters of interest typically requires extensive recalculations.

Optimistically, one may hope that a (numerical or analytical) estimate of the
infidelity of the logical noise under a probabilistic Pauli channel generalizes
directly to general logical noise. However, even quantifying the error becomes
more complicated for more general noise. The ``error rate'' due to a noise
process $\mc N$ acting on a $m$-level system is often experimentally quantified
via the average gate infidelity to the identity (hereafter the infidelity)
\begin{align} r(\mc N) = 1- \int {\rm d}\psi
\langle\psi|\mc N(|\psi\rangle\!\langle\psi|)|\psi\rangle
\end{align}
because it can be efficiently estimated via randomized
benchmarking~\cite{Emerson2005, Emerson2007, Dankert2009, Knill2008,
Magesan2011}. However, theoreticians often report rigorous bounds on the
performance of a quantum error-correcting code or a circuit in terms of the
diamond distance to the identity (hereafter the diamond
distance)~\cite{Kitaev1997}
\begin{align}\label{eq:diamond_def}
\epsilon(\mc N)
= \sup_{\psi} \tfrac{1}{2}\| \left[\mc
N\otimes\mc I_m-\mc I_{m^2}\right]\!(\psi)\|_1
\end{align}  where $\|A\|_1 =
\sqrt{\textrm{Tr} A^{\dagger} A}$ and the maximization is over all
$m^2$-dimensional pure states (to account for the error introduced when acting
on entangled states).

The infidelity and diamond distance are related via the bounds~\cite{Beigi2011,
Wallman2014}
\begin{align}\label{eq:fidelity_to_worst} r(\mc N) (1+m^{-1}) \leq
\epsilon(\mc N) \leq \sqrt{m(m+1)r(\mc N)}.
\end{align}
which scale optimally with respect to $r$ and $m$~\cite{Sanders2015}. For unitary noise, $\epsilon(\mc N)$ scales as $\sqrt{r(\mc N)}$, though it does not necessarily saturate the upper bound of \cref{eq:fidelity_to_worst}; this scaling follows from the magnitude of the coherent (non-Pauli) part of the noise~\cite{Kueng2016}. Pauli
noise saturates the lower bound of \cref{eq:fidelity_to_worst} and the effect
of coherent noise is often assumed to be negligible, so that experimental
infidelities are often compared to diamond distance targets to determine
whether fault tolerance is possible~\cite{Sanders2015}. However, even if
coherent errors make a negligible contribution to the infidelity, they can
dominate the diamond norm~\cite{Wallman2016a}. Because of this uncertainty
about how to quantify errors effectively, it is unclear what figure of merit
recovery operations should optimize and how to quantify the logical error rate~\cite{Gutierrez2016,Chamberland2017,Pavithran2017}.

Previous studies have shown that the contribution to the logical noise from
the coherent part of the physical noise decays exponentially as a function of
code distance~\cite{Fern2006}, although the decay rate was only given as an abstract property of the noise map. Recently, the decay rate was analyzed for specific noise models in
the repetition code~\cite{Greenbaum2018}.

In this paper, we directly relate the decay rate of coherent terms at the logical level of a general
stabilizer code to the infidelity of the physical noise of a general local
noise process, which can be estimated by randomized benchmarking. Further, we
give physical motivation for the decoherence of errors with increasing code
distance by relating the scaling of errors to projective syndrome
measurements. We demonstrate that—even without applying recovery operations—encoding a system in a quantum error-correcting code and measuring error syndromes decoheres errors, that is, causes rapid convergence toward probabilistic Pauli errors. To isolate the contribution from local noise, we assume that there is no other contributing noise. That is, encoding, syndrome measurements, recovery operations, and decoding are all assumed to be noiseless.

Our results show that the effective logical noise is well characterized by the logical infidelity. This provides a rigorous justification for choosing recovery maps to independently optimize the logical fidelity per syndrome (instead of, for example, optimizing the diamond norm of the logical noise averaged over all syndromes). Complementary results on the scaling of the diamond distance with quantum error correction protocols were independently obtained in ref. \cite{Huang2018}.

The paper is structured as follows. We first introduce Markovian noise processes
and review the process matrix formalism, a convenient representation of quantum
channels (not to be confused with the $\chi$ matrix representation). We then give an expression for the infidelity in terms of this
representation and discuss the implications and bounds on the entries of a
process matrix in terms of its infidelity. Next, we introduce stabilizer codes
and, using the aforementioned bounds, discuss the behavior of the effective
logical noise of an encoded state after syndrome measurements with and without
the application of recovery operations in terms of the physical infidelity of
the qubits. We conclude by discussing some implications of our work and discuss
how our results relate to existing results showing coherent errors at the
logical level.

\section{Markovian Noise Processes}

We represent quantum states and measurements of a $m$-dimensional system by
vectors as follows. Let $\{e_j:j\in\bb Z_m\}$ be the canonical basis of
$\bb C^{m^2}$ and $\bb B$ be an arbitrary trace-orthonormal basis of
$\bb C^{m\times m}$ respectively, that is, $\tr (B_j\ct B_k) = \delta_{j,k}$ for
all $B_j,B_k\in \bb B$. We will generally choose $\bb B$ to be the set of
normalized (physical or logical) Pauli operators, $\bb P =
\{I_2,X,Y,Z\}/\sqrt{2}$, or tensor products thereof. We define a map
$|.\rangle\!\rangle:\bb C^{m\times m}\to\bb C^{m^2}$ by setting $|B_j\rangle\!\rangle\to e_j$ for all
$B_j\in\bb B$ and extending to a linear map, so that
\begin{align}\label{eq:state}
|M\rangle\!\rangle &= \sum_j \tr(B_j\ct M) e_j.
\end{align}
Defining $\langle\!\langle M| = |M\rangle\!\rangle\ct$, we have
\begin{align}
\langle\!\langle M|N\rangle\!\rangle = \tr(M\ct N).
\end{align}

A Markovian noise process is a linear map $\mc N$ that maps valid quantum states
of one system to valid quantum states of another system, and so is completely
positive and trace preserving (CPTP). Let $\bb B\ini$ and $\bb B\out$ be
trace-orthonormal bases for the input and output systems respectively. Then
\begin{align}\label{eq:processMatrix}
|\mc N(M)\rangle\!\rangle &= \sum_{B\in\bb B\ini} |\mc N(B)\rangle\!\rangle\!\langle\!\langle B|M\rangle\!\rangle \notag\\
&= \mc N|M\rangle\!\rangle,
\end{align}
where we abuse notation slightly by using $\mc N$ to denote both an abstract map
and its matrix representation $\sum_{B\in\bb B\ini} |\mc N(B)\rangle\!\rangle\!\langle\!\langle B|$. Note
that $|\mc N(B)\rangle\!\rangle$ is a state of the output system and so is expanded
relative to $\bb B\out$ via \cref{eq:state}. The composition of two channels is
then given by the standard matrix product of the process matrices.

The average infidelity of a single-qubit noise process $\mc N$ with the identity in terms of
process matrices is~\cite{Kimmel2014}
\begin{align}
r=\frac{\textrm{Tr}[\mc I-\mc N]}{6}. \label{eq:procinfidel}
\end{align}
The infidelity only captures the effects of the Pauli part of the noise, that
is, the diagonal part, whereas the disconnect between the infidelity and the
diamond norm in \cref{eq:fidelity_to_worst} for non-Pauli noise is due to the
off diagonal terms, which we call the coherent part of the noise.

Setting $B_0 = I_2/\sqrt{2}$ and defining the single-qubit error matrix
$E\equiv \lvert I_4-\mc N\rvert$, we have the following bounds on the matrix
entries $E_{\sigma,\tau} = \langle\!\langle \sigma|E|\tau\rangle\!\rangle$ of $E$ in terms of the
infidelity.

\begin{lem}\label{lem:element_bounds} For any single-qubit Markovian noise process with
infidelity $r$,
\begin{subequations}

\begin{align}
E_{\sigma_0,\sigma} &= 0 \label{eq:TP}\\
E_{\sigma,\sigma_0} &\leq 3r \label{eq:nonUnital}\\
E_{\sigma,\sigma} &\leq 3r \label{eq:diagonal}\\
E_{\sigma,\tau} &\leq \sqrt{6r} \label{eq:offDiagonal}
\end{align}
\end{subequations}
for all $\sigma,\tau\in\vec{\sigma} = {I,X,Y,Z}/\sqrt{2}$.
\end{lem}

\begin{proof}
\Cref{eq:TP} follows directly from the trace-preserving condition.
\Cref{eq:nonUnital}  was proven in \cite[Prop. 12]{Wallman2014}.
To prove \cref{eq:diagonal}, note that the Pauli twirl of $\mc N$,

\begin{align}
\frac{1}{4}\sum_{P\in\{I,X,Y,Z\}} \mc P\mc N\mc P
\end{align}
where $\mc P$ denotes the channel that acts via conjugation by $P$, is a valid
channel whose process matrix is the diagonal part of $\mc N$ whose singular values
are consequently the diagonal entries. We can then write $E_{\sigma,\sigma} =
a_\sigma r$ ~\cite{Ruskai2002} where the $a_\sigma$ must satisfy

\begin{align}
(a_\sigma - a_\tau)^2 \leq a_\nu^2
\end{align}
for all permutations $\{\sigma,\tau,\nu\}$ of $\sigma\backslash\{\sigma_0\}$ in
order for the map to be CPTP~\cite[eq. (63)]{Wallman2014} and must add to $6$,
by \cref{eq:procinfidel}, as $\mc N$ has infidelity $r$.

\Cref{eq:offDiagonal} holds as the Euclidean norm of any column of $\mc N_{\rm u}$
is upper-bounded by 1 where $\mc N_{\rm u}$ is the unital block obtained by
deleting the first row and column of $\mc N$~\cite{Ruskai2002}. Note that the term in
the square root was only kept to $\mc O{(r)}$; an $r^2$ term was dropped, reducing
the inequality from $E_{\sigma,\tau} \leq \sqrt{6r-9r^2}$. This convention will
be followed for the remainder of the paper. This bound can be tightened further
by considering unitarity~\cite{Wallman2015}.
\end{proof}

\section{Stabilizer Codes}

We now review stabilizer codes; for more details, see, for example,
Ref.~\cite{Gottesman2010}. Let $[A, B] = AB - BA$ and $\{A,B\} = AB + BA$. An
$n$-qubit Pauli operator $P$ is the tensor product of $n$ single-qubit Pauli
operators, and the weight $w(P)$ of a Pauli operator $P$ is the number of qubits
$P$ acts on nontrivially. An $[\![n,k,d]\!]$ stabilizer code encodes $k$ logical qubits in $n$ physical qubits and is distance $d$; it is defined by an Abelian group
$\bb S\not\ni -I$ of $2^{n-k}$ $n$-qubit Pauli operators, which can be described
by a set of generators $g_1,\ldots,g_{n-k}$. We can define a set of $2^{n-k}$
mutually orthogonal projectors
\begin{align}\label{eq:projectors}
 \Pi_s = \prod_{j=1}^{n-k} \frac{1}{2}(I+(-1)^{s_j} g_j),
\end{align}
where $s_j$ is the $j$th entry of the syndrome, $s$, and the code space is the support of
$\Pi_0$. An error is detectable if it maps the support of $\Pi_0$ outside of
$\Pi_0$ and has no effect if it acts trivially on $\Pi_0$, that is, if it is in
$\bb S$. The distance of the code is the minimal Pauli weight of an undetectable
error that acts nontrivially on $\Pi_0$. For each error syndrome
$s\in\bb Z_2^{n-k}$ we can find a Pauli operator $R_s$ satisfying $R_s \Pi_s R_s
= \Pi_0$ which corrects the error.

We can find a set of operators $\{\overline{X}_j, \overline{Z}_j: j=1, \ldots,
k\}$ such that for all $S\in\bb S$ and $j\neq k$,
\begin{align}
[\overline{X}_j,S] &= [\overline{Z}_j, S] =0 \notag\\ [\overline{X}_j,
\overline{X}_k] &= [\overline{X}_j, \overline{Z}_k] = [\overline{Z}_j,
\overline{Z}_k] = 0 \notag\\ \overline{X}_j \overline{Z}_j &= - \overline{Z}_j
\overline{X}_j.
\end{align}
Let $\bb L$ be the projective group generated by
$\{\overline{X}_j, \overline{Z}_j: j=1, \ldots, k\}$.
Then $2^{-k/2}\bb L\Pi_0$ is a trace-orthonormal set of operators that span the code space.
Therefore any operator $\overline{\rho}$ in the code space can be written as
\begin{align} \overline{\rho} = 2^{-k} \sum_{L\in\bb L} \tr(L\Pi_0 \overline{\rho}) L
\Pi_0.
\end{align}

\section{Effective Noise Under Error Correction}

We now prove that, even with bad decoders (or no correction), encoding in an
error correcting code decoheres local errors.

For ideal encoding and correction operations, preparing an initial state in the
code space, applying a general local $n$-qubit noise process $\mc N=\mc N^{(1)}\otimes\mc N^{(2)}\otimes...\otimes\mc N^{(n)}$, and performing a syndrome
measurement with the outcome $s$ maps the system from the support of $\Pi_0$ to
that of $\Pi_s$. Let $p(s)$ be the probability of observing the syndrome $s$,
which will generally depend upon the input state. Then by
\cref{eq:processMatrix} the effective noise map from $\Pi_0$ to $\Pi_s$ is
\begin{align}\label{eq:average_process_matrix}
\overline{\mc N}(s)_{L,L'}&=
\frac{\langle\!\langle L\Pi_s|\mc N|L'\Pi_0\rangle\!\rangle}{p(s) 2^k} ,
\end{align}
where the factor of $2^{-k}$ comes from the normalization of $\bb L\Pi_s$~\cite{Rahn2002}. Note
that it is conventional to apply a ``pure error''~\cite{Poulin2006} to map back
to the code space. We omit this step to highlight the fact that syndrome
measurements alone decohere the noise.

\begin{thm}\label{thm:off_diagonal_logical}
For any $[\![n,k,d]\!]$ stabilizer code, the average off diagonal elements of
the logical noise under a local noise process $\mc N = \bigotimes_{j=1}^n \mc N^{(j)}$
scales as
\begin{align}
\sum_s p(s)\overline{\mc N}(s)_{L,L'} \in \mc O(r^{d/2}) \mbox{ as }
r\to 0
\end{align}
where $r = \max_j r(\mc N^{(j)})$.
\end{thm}

\begin{proof}
By \cref{eq:projectors}, \cref{eq:average_process_matrix} can be rewritten as
\begin{align}\label{eq:exact}
\overline{\mc N}(s)_{L,L'}&=
\sum_{S,S'\in\bb S}\frac{\phi(S|s) \langle\!\langle L S|\mc N|L'S'\rangle\!\rangle}{p(s) 2^{2n-k}} ,
\end{align}
where $\phi(S|s)$ is the sign of $S$ in the expansion of \cref{eq:projectors}.
As $\mc N$ and the stabilizers are all tensor products, terms of the form $\langle\!\langle L S|\mc N|L'S'\rangle\!\rangle$ can be factorized.
However, this introduces a subtlety as $LS$ may be a phase multiple of an element of $\{I,X,Y,Z\}^{\otimes n}$, which needs to be accounted for when factoring the tensor product.
Let $\chi(A)\in\{\pm, \pm i\}$ be the phase multiple of $A$ relative to its representative element $A'$ in the projective Pauli group $\{I,X,Y,Z\}^{\otimes n}$ so that $A = \chi(A)A'$.
Note that we can ignore the $\pm i$ case as all operators under consideration are Hermitian.
Then, using $\mc N_{P,Q} = \langle\!\langle P|\mc N^{(j)}|Q\rangle\!\rangle/2$ for $P,Q\in\{I,X,Y,Z\}$,
\begin{align}\label{eq:exact2}
\overline{\mc N}(s)_{L,L'} &= \sum_{S,S'\in\bb S}\frac{\phi(S|s)\chi(LS)\chi(L'S')}{p(s) 2^{n-k}} \prod_{j=1}^n
\mc N^{(j)}_{L_jS_j, L'_jS'_j}.
\end{align}
By the definition of the code distance, $SL$ and $S'L'$ differ
on at least $d$ qubits for $S\in \bb S L$, $S'\in\bb S L'$ and $L\neq L'$.
Therefore for any $L\neq L'$, each term on the right-hand side of
\cref{eq:exact2} is in $\mc O(r^{d/2})$ by \cref{lem:element_bounds} after
syndrome measurements. Averaging over the syndromes cancels the $p(s)$ in the
denominator.
\end{proof}

Intuitively, syndrome measurements decohere errors because the act of measuring
projects out any Pauli in the expansion of the output state that is not of the
form $LS$, thus removing the components of the output state corresponding to the additional Pauli operators introduced by coherent noise.

In \cref{thm:off_diagonal_logical}, we proved that any errors are suppressed
exponentially with the code distance. To conclude that the noise is
decohered, we need to show that the off diagonals of the logical error matrix,
$E$, do not scale as the square root of the diagonals, so that the ratio of the off diagonals to diagonals decreases with code distance (ie the ratio of the off diagonal elements to the diagonal elements of the logical noise is less than the corresponding ratio for the physical noise). To see that this holds,
at least for typical noise in nondegenerate stabilizer codes, note that
\cref{eq:exact} is linear in $\mc N$. Writing $\mc N = \sum_{x\subset \mathbb{Z}_n}
E(x)$ where $E(x)$ is an error that only acts nontrivially on qubits in $x$ and
$E(\emptyset) = I$,
\begin{align}
\overline{\mc N}(s)_{L,L'}=
\sum_{S,S'\in\bb S, x\subset\bb{Z}_n}
&\frac{\phi(S|s)\chi(SL)\chi(S'L')}{p(s) 2^{n-k}}\notag\\
 &\times\prod_{j\in x} E(x)^{(j)}_{L_jS_j, L'_jS'_j}.
\end{align}

For a nondegenerate distance $d$ stabilizer code, there exists some set $x$ of
at most $\lceil d/2 \rceil$ qubits such that $E(x)$ cannot be corrected, that is,
canceled out when averaged over syndromes.
This set contributes a term $\sum_{S\in\mathbb{S}}\prod_{j\in x} E(x)^{(j)}_{L_jS_j, L_jS_j}$. By reducing
the generators so that at most one generator acts nontrivially as $\sigma$ on
each $j\in x$ for each $\sigma\in\vec{\sigma}$, we can find some stabilizer such
that $L_j S_j\neq \sigma_0$ for all $j\in x$. Let
\begin{align}
r' = \min_{j,\sigma\in\vec{\sigma}} E(x)^{(j)}_{\sigma,\sigma},
\end{align}
which will be $\mc O(r)$ for typical noise. Then $x$ contributes a term that scales
as at least $r'^{|x|}$ to the effective logical error and so the
logical infidelity scales as ${r'}^{\lceil d/2 \rceil}$ or worse, so that the
off diagonals are, at worst, proportional to the diagonals of the logical error
matrix.

As $d$ increases, the scaling described above causes the effective logical
noise to become progressively less coherent so that the Pauli twirl approximation captures the logical noise more effectively. However, due to contributions from the coherent part of the physical noise to the Pauli part of the logical noise, approximating the physical noise as Pauli in order to calculate the logical noise produces inaccurate results as observed previously~\cite{Gutierrez2016,Greenbaum2018}. Ref. \cite{Greenbaum2018} demonstrated that the coherent contribution dominates the Pauli part of the logical noise after many rounds of error correction. We now apply our bounds on the scaling of errors to a more general analysis of error accumulation in a scheme with rounds of error correction. The effective logical noise after $h$ rounds of error correction is

\begin{align}
	(I-\overline{E})^h\approx I-h\overline{E}+{h\choose 2}\overline{E}^2,
\end{align}
where we have taken a binomial expansion to second order in $\overline{E}$. Assuming typical noise, the off diagonals of $\overline{E}$ scale at worst as $\mc O(r^{(d+1)/2})$, and the diagonals as $\mc O(r^{d/2})$. When the noise  is Pauli, the effective logical noise on the diagonal after $h$ rounds of error correction will be at worst

\begin{align}
	(I-\overline{E})^h_{\sigma,\sigma}\approx 1-\mc O(hr^{(d+1)/2})+\mc O(h^2r^{d+1}).
\end{align}

If coherent noise is present,

\begin{align}
(I-\overline{E})^h_{\sigma,\sigma}\approx 1-\mc O(hr^{(d+1)/2})+\mc O(h^2r^{d}).
\end{align}

Taking the ratio of the first and second order terms, quadratic errors start to accumulate from Pauli noise at $h_{P}\approx 1/r^{(d+1)/2}$ and from coherent noise at $h_c\approx 1/r^{(d-1)/2}$. The coherent noise begins to dominate the Pauli part of the effective logical noise occurs at $h_{crit}\approx 1/r$, independent of the code distance. This critical value is consistent with the value observed in ref. \cite{Greenbaum2018} of $1/\epsilon^2$, where $\epsilon$ is the angle of rotation about the x-axis, and we note that all of our observations hold in their specific case when we replace $r$ in our results with $\sqrt{\epsilon}$, as that is how the specified noise scales relative to our \cref{lem:element_bounds}. Because the off diagonal terms and diagonal terms produce the same scaling in a worst-case analysis with coherent noise, the ratio of off diagonal to diagonal errors is independent of the number of rounds of error correction in the worst-case scaling of typical noise.

\section{Conclusion}

In this paper, we have shown that for generic local noise, coherent errors are
decohered by syndrome measurements in error correcting stabilizer codes.
Consequently, error rates in logical circuits are well quantified by the logical
infidelity. Therefore it is appropriate to choose recovery operators to optimize
the logical fidelity, instead of other measures such as the diamond norm. This
dramatically simplifies the process of selecting recovery operators for general
noise because the fidelity is a linear function of quantum channels and so we
can optimize the fidelity of the logical noise for each syndrome independently,
as noted in \cite{Chamberland2017}. By contrast, if we tried to optimize the
diamond norm of the average logical noise, we would have to simultaneously
optimize all recovery operators.

While we have only explicitly considered independent errors, note that our
arguments apply directly to correlated errors of the form
\begin{align}
\mathcal{N} = \sum_\alpha p_\alpha \bigotimes_{j=1}^n \mathcal{N}^{(\alpha,j)}
\end{align}
by linearity. The only nontrivial issue is identifying a scaling parameter akin
to the single-qubit infidelity.

Previous results have demonstrated significant logical coherent
errors~\cite{Fern2006, Gutierrez2016}, namely, off diagonals that scale as
$r^{3/2}$ compared to diagonals that scale as $r^2$. However, these results were
all for distance 3 codes and are consistent with our results as for such codes,
$\lceil d/2\rceil = 2$ giving diagonals that scale as ${r'}^{2}$ and
off diagonals that scale as $r^{3/2}$ by \cref{thm:off_diagonal_logical}.
Numerically, significant discrepancies between the logical diamond norm error
with and without Pauli twirling (which removes the coherent part of the noise) at the physical level have been observed for high distance surface
codes~\cite{Darmawan2017} (up to distance 10).
These discrepancies have been interpreted as
suggesting significant logical coherent errors~\cite{Greenbaum2018}. Our
results show that these discrepancies are almost entirely due to contributions to the logical infidelity from the coherent part (ie off diagonals) of the physical noise\footnote{Recall that the logical infidelity depends only on the diagonal part of the process matrix so that these discrepancies are almost entirely due to contributions from the coherent part of the physical noise to the logical diagonals.}, though for a specific syndrome and noise model, the effective logical noise may appear coherent. That is, the effective logical noise is generically very close to a Pauli channel on average, however, it may not be the Pauli channel one would predict from the Pauli twirl of the physical noise.

\section{Acknowledgements}

This research was supported by the Canadian federal and Ontario provincial
governments through an NSERC CGS-M and an Ontario Graduate Scholarship. This
research was undertaken thanks in part to funding from TQT, CIFAR, the
Government of Ontario, and the Government of Canada through CFREF, NSERC and
Industry Canada. MG and KRB were supported by the ODNI-IARPA LogiQ program.

\bibliography{library}

\end{document}